\documentclass[12pt]{article}

\usepackage[hmargin=1.3in,vmargin=1.3in]{geometry}
\usepackage{bbding}
\usepackage{mathrsfs}
\usepackage{}
\usepackage{amsfonts}
\usepackage{multirow}
\usepackage{amsfonts,amssymb,amsmath,amsthm,bm}
\usepackage[T1]{fontenc}
\usepackage[utf8]{inputenc}
\usepackage{authblk}
\usepackage[colorlinks,
            linkcolor=blue,
            anchorcolor=blue,
            citecolor=blue]{hyperref}

\newtheorem{theorem}{Theorem}
\newtheorem{example}{Example}
\newtheorem{definition}{Definition}
\newtheorem{remark}{Remark}

\newtheorem{lemma}{Lemma}
\newtheorem{proposition}{Proposition}

\begin{document}

\title{Singleton-Optimal LRCs and Perfect LRCs via
Cyclic and Constacyclic Codes\footnote{This paper was presented in part at the 2021 IEEE International Symposium on Information Theory (ISIT) \cite{FCXF21}}}
\author[,1,2,3]{Weijun Fang\thanks{Corresponding author: fwj@sdu.edu.cn}}
\author[4]{Fang-Wei Fu}
\author[5]{Bin Chen}
\author[6]{Shu-Tao Xia}
\affil[1]{Key Laboratory of Cryptologic Technology and Information Security, Ministry of Education, Shandong University, Qingdao, China}
\affil[2]{School of Cyber Science and Technology, Shandong University, Qingdao, China}
\affil[3]{Quancheng Laboratory, Jinan, China}
\affil[4]{Chern Institute of Mathematics and LPMC, Nankai University, Tianjin, China}
\affil[5]{Department of Computer Science and Technology, Harbin Institute of Technology, Shenzhen, China}
\affil[6]{Tsinghua Shenzhen International Graduate School, Tsinghua University, Shenzhen, China}

\date{}
\maketitle
\begin{abstract}
Locally repairable codes (LRCs) have emerged as an important coding scheme in distributed storage systems (DSSs) with relatively low repair cost by accessing fewer non-failure nodes. Theoretical bounds and optimal constructions of LRCs have been widely investigated. Optimal LRCs via cyclic and constacyclic codes
provide significant benefit of elegant algebraic structure and efficient encoding procedure.  In this paper, we continue to consider the constructions of optimal LRCs via cyclic and constacyclic codes with long code length. Specifically, we first obtain two classes of $q$-ary cyclic Singleton-optimal $(n, k, d=6;r=2)$-LRCs with length $n=3(q+1)$ when $3 \mid (q-1)$ and $q$ is even, and length $n=\frac{3}{2}(q+1)$ when $3 \mid (q-1)$ and $q \equiv 1(\bmod~4)$, respectively. To the best of our knowledge, this is the first construction of $q$-ary cyclic Singleton-optimal  LRCs with length $n>q+1$ and minimum distance $d \geq 5$. On the other hand, an LRC acheiving the Hamming-type bound is called a perfect LRC. By using cyclic and constacyclic codes, we construct two new families of $q$-ary perfect LRCs with length $n=\frac{q^m-1}{q-1}$, minimum distance $d=5$ and locality $r=2$.

\end{abstract}


\maketitle

\section{Introduction}
Distributed storage systems (DSSs) have become a popular way to store large amounts of data in modern data center, e.g., Windows \cite{WIN12} and Facebook \cite{FB13}. By distributing huge data over multiple storage nodes, DSSs can not only reduce the storage cost, but also improve
the reliability and scalability. However, node failures become a norm in DSSs rather than an exception. Thus erasure codes that optimize the repair bandwidth and repair degree (the number of nodes used in a repairing process) for recovering a failed node have been adopted in DSSs.  In \cite{GHSY12}, Gopalan \emph{et al.} first introduced the locality in an erasure code to reduce the repair degree, that is, a failed node can be recovered by accessing few available nodes.  Let $C$ be a $q$-ary $[n, k, d]$ linear code, which is just a subspace of $\mathbb{F}_{q}^{n}$ of dimension  $k$ with minimum Hamming distance $d$. Formally,  $C$ is called a \emph{locally repairable code} (LRC) with \emph{locality} $r$ if for any $i\in[n]\triangleq\{1, 2,\dots, n\}$, there exists a subset $R_i\subset [n]$ with $i \in R_i$ and $|R_{i}| \leq r+1$ such that the $i$-th symbol $c_i$ can be recovered by $\{c_j\}_{j\in R_i\setminus \{i\}}$, i.e., $c_i$ can be represented as a linear combination of $\{c_j\}_{j\in R_i\setminus \{i\}}$. $R_i$ is called the \emph{repair group} of the $i$-th symbol. An LRC is called to have \emph{disjoint local repair groups} if for any $1\leq i \neq j \leq n$, we have $|R_{i}|=r+1$ and either $R_{i} \cap R_{j} =\emptyset$ or $R_{i} = R_{j}$. There are some restrictions among the parameters $n, k, d, r$. One of these bounds is the famous Singleton-type bound derived in \cite{GHSY12}, which said that for any $q$-ary $[n, k, d]$-LRCs with locality $r$,
\begin{equation}
\label{singleton}
d\leq n-k+2-\left\lceil k/r\right\rceil.
\end{equation}

Note that when $r=k$, \eqref{singleton} reduces to the classical Singleton bound. 
\begin{definition}[Singleton-optimal LRCs]
A $q$-ary $[n, k, d]$ linear code $C$ with locality $r$ is called a \emph{Singleton-optimal} $(n, k, d; r)$-LRC if it achieves the bound \eqref{singleton} with equality.
\end{definition}

From the Singleton-type bound \eqref{singleton}, a Singleton-optimal $(n, k, d;r)$-LRC has to satisfy that $\frac{k}{n}=\frac{r}{r+1}-\frac{d-2}{n}$ (suppose $r \mid k$). Thus for fixed $d$ and $r$, the longer code length of a Singleton-optimal LRC leads to larger code rate. The construction of Singleton-optimal LRCs with long code length has been a hot topic in coding theory in recent years (see \cite{CFXF19,CFXHF20,CXHF18,GXY19,J19,JMX19,LMX19,LXY19,TB14,TBGC15,XY18}). In \cite{TB14}, Tamo and Barg  first gave a breakthrough construction of Singleton-optimal LRCs via subcodes of Reed-Solomon codes whose code length can go up to the alphabet size.
Due to their efficient encoding and decoding performance, cyclic codes and constacyclic codes have been widely used in the constructions of optimal LRCs (see \cite{CFXF19, CXHF18,LXY19,TBGC15}).  For $n \mid (q+1)$, by  using the techniques of cyclic and constacyclic codes, Chen \emph{et al.} \cite{CFXF19,CXHF18} have constructed the $q$-ary Singleton-optimal LRCs for all possible parameters. In \cite{JMX19}, Jin \emph{et al.}  also presented a construction of $q$-ary optimal LRCs with length $n=q+1$ via the automorphism groups of rational function fields.  For minimum distance $d=3$ or $4$,  Luo \emph{et al.} \cite{LXY19} proposed a construction of cyclic optimal LRCs with unbounded length. For minimum distance $d \geq 5$, Guruswami \emph{et al.} \cite{GXY19} proved that the code length of a $q$-ary optimal LRC must be at most roughly $O(dq^{3})$ and presented an algorithmic construction of $q$-ary optimal LRCs of length $\Omega_{d, r}(q^{1+1/\lfloor\frac{d-3}{2}\rfloor})$ for $r \geq d-2$. By employing algebraic structures of elliptic curves, Li \emph{et al.} \cite{LMX19} constructed a family of $q$-ary optimal LRCs with length up to $q+2\sqrt{q}$ and some particular localities. For $d=5, 6$ and $r \geq d-2$, Jin \cite{J19} presented an explicit construction of $q$-ary optimal LRCs of length $\Omega_{r}(q^{2})$ via binary constant weight codes. In \cite{XY18}, Xing and Yuan generalized Jin's results and presented a construction of $q$-ary optimal LRCs for general $d \geq 7$ via hypergraph theory.  In \cite{CFXHF20}, Chen \emph{et al.} considered the case of $r < d-2$.  In particular, for $d=6$ and $r=2$, the authors established an equivalent connection between the existence of Singleton-optimal LRCs and the existence of  lines in finite projective plane with certain properties. And they also provided some new constructions and improved bounds on Singleton-optimal LRCs. However, as we have seen, for $d \geq 5$, the longest code length of known $q$-ary cyclic Singleton-optimal LRCs is $q+1$ until now. So it is quite interesting to construct $q$-ary cyclic Singleton-optimal LRCs with length $n >q+1$ and $d \geq 5$.  

On the other hand, the Singleton-type bound \eqref{singleton} is not always achievable in some cases (see \cite{SDYL14,TPD16}), and there are only a few  Singleton-optimal LRCs over binary or ternary fields (see \cite{Hao,ternary}). Some other theoretical bounds on the parameters of LRCs were also investigated in the literature.  Cadambe and Mazumdar \cite{CM15} first derived a theoretical bound of LRCs depending on the field size, which is known as \emph{C-M} bound. Some well-known binary codes, such as Hamming codes, Simplex codes etc., are optimal with respect to the \emph{C-M} bound (see \cite{HYUS16}). In \cite{ABHMT18},  Agarwal \emph{et al.} presented several new
combinatorial bounds on LRC codes including the locality-aware
sphere packing and Plotkin bounds. Recently, Tebbi \emph{et al.} \cite{TCS} also obtained a linear programming bound for LRCs with arbitrary code parameters. However, it is not easy to give an explicit form of the \emph{C-M} bound or the bound in \cite{TCS} in general. In \cite{HXSCFY20}, Hao \emph{et al.} generalized the C-M bound and obtained a Griesmer-type bound.  In \cite{WZL19}, Wang \emph{et al.} derived a new explicit bound on the size of LRCs via a \emph{sphere-packing based approach}. Precisely, they proved that
 \begin{equation}\label{2}
   A_{q}(n,d;r) \leq \frac{q^{\dim(\mathcal{V})}}{B_{\mathcal{V}}\left(\left\lfloor(d-1)/2\right\rfloor\right)},
 \end{equation}
 where $A_{q}(n,d;r)$ is the maximal number of codewords of a $q$-ary LRC of length $n$, minimum distance $d$, locality $r$ and with disjoint local repair groups, $\mathcal{V}$ is the $\mathbb{F}_q$-linear space containing all of $(n,d;r)$-LRCs with disjoint local repair groups and  $B_{\mathcal{V}}(\lfloor\frac{d-1}{2}\rfloor)\triangleq|\{\bm v \in \mathcal{V}: wt(\bm{v}) \leq \lfloor\frac{d-1}{2}\rfloor\}|.$ Note that when $\mathcal{V}$ is the whole space $\mathbb{F}_{q}^{n}$, the bound \eqref{2} reduces to the well-known classical Hamming bound. On the other hand, as the authors mentioned in \cite{WZL19}, they obtained the bound \eqref{2} by using the essential idea of the Hamming bound, not just its
expression. Thus, in this paper, we regard bound \eqref{2} as a \emph{Hamming-type bound} for LRCs with disjoint local repair groups. Moreover, in \cite[Theorem 11]{WZL19}, taking $\mathcal{V}$ as the dual space of the locality rows of  the parity-check matrix, they proved that \[B_{\mathcal{V}}(\lfloor\frac{d-1}{2}\rfloor)=\sum_{0 \leq i_{1}+\dots+i_{\ell}\leq \lfloor\frac{d-1}{2}\rfloor}\prod_{j=1}^{\ell}\beta(r, i_{j}),\] where $\beta(r, i)=\frac{1}{q}\big((q-1)^{i}+(-1)^{i}(q-1)\big)\binom{r+1} {i}$. Thus we can rewrite bound \eqref{2} as a more explicit form, i.e., 
\begin{equation}\label{3}
  q^{k} \leq \frac{q^{\frac{rn}{r+1}}}{\sum\limits_{0 \leq i_{1}+\dots+i_{\ell}\leq \lfloor\frac{d-1}{2}\rfloor}\prod_{j=1}^{\ell}\beta(r, i_{j})}.
\end{equation}
It can be seen that for $d=5$ or $6$, bound \eqref{3} becomes
\begin{equation}\label{4}
q^{k} \leq \frac{q^{\frac{r}{r+1}n}}{\frac{rn}{2}(q-1)+1}.
\end{equation}
\begin{definition}[Perfect LRCs, \cite{FCXF20-2,Fang22}]
A $q$-ary $(n, k, d;r)$-LRC with disjoint local repair groups is called a \emph{perfect LRC} if it  achieves the bound \eqref{2} (or \eqref{3})  with equality.
\end{definition}

 As we known, Hamming codes, Binary and Ternary Golay codes are the all nontrivial perfect linear codes in classical coding theory, so it is expected to determine  perfect LRCs for all possible parameters. From the Hamming-type bound \eqref{3}, we know that if there exists a $q$-ary $(n, k, d;r)$-perfect LRC with disjoint repair groups, then  $\sum\limits_{0 \leq i_{1}+\dots+i_{\ell}\leq \lfloor\frac{d-1}{2}\rfloor}\prod_{j=1}^{\ell}\beta(r, i_{j})$ should be exactly a power of $q$. So it is quite difficult  to construct perfect LRCs and there are few works on constructions of perfect LRCs.
In \cite{GC14}, Goparaju \emph{et al.} have presented a family of binary perfect $(n=2^{u}-1, k=\frac{2}{3}n-u, d=6; r=2)$-LRCs for even $u$, which is  the only one known class of cyclic perfect LRCs until now. Some  new classes of perfect LRCs have been presented in \cite{FCXF20-2,Fang22} via the techniques of finite geometry and finite fields.



 In this paper, firstly, we will give two new constructions of cyclic Singleton-optimal LRCs with longer length. Precisely, for $3 \mid (q-1)$ and $q$ is even, we give a construction of $q$-ary cyclic optimal LRCs with length $n=3(q+1)$, minimum distance $d=6$ and locality $r=2$; For $3 \mid (q-1)$ and $q \equiv 1(\bmod~4)$, we also obtain a family of $q$-ary cyclic Singleton-optimal LRCs with length $n=\frac{3}{2}(q+1)$, minimum distance $d=6$ and locality $r=2$. To the best of our knowledge, this is the first construction of $q$-ary cyclic Singleton-optimal  LRCs with length $n>q+1$ and minimum distance $d \geq 5$. Secondly, for even $q \geq 4$ and even $m$ with $3 \mid (q+1)$ and $\gcd(m,q-1)=1$, we construct a family of cyclic $(n=\frac{q^{m}-1}{q-1}, k=\frac{2}{3}n-m, d=5; r=2)$-perfect LRCs. Moreover, for general $q$ with $3 \mid (q+1)$ and even $m$, we construct a family of constacyclic $(n=\frac{q^{m}-1}{q-1}, k=\frac{2}{3}n-m, d=5; r=2)$-perfect LRCs. To the best of our knowledge, this is the first construction of $q$-ary (consta)cyclic perfect LRCs for general $q$.

The rest of this paper is organized as follows. In Section 2, we briefly review some basic results about LRCs and (consta)cyclic codes. In Section 3, we provide our new constructions of cyclic Singleton-optimal LRCs. In Section 4, we give two new constructions of cyclic and constacyclic perfect LRCs, respectively.  We conclude this paper in Section 5.

\section{Preliminaries}
In this section, we briefly review some basic results about cyclic codes, constacyclic codes and LRCs.

Suppose $q$ is a prime power and $\mathbb{F}_{q}$ is a finite field with $q$ elements. Denote $[n] \triangleq \{1,2,\dots, n\}$. For a vector $\bm{v}=(v_{1}, v_{2}, \cdots, v_{n}) \in \mathbb{F}_{q}^{n}$, the support of $\bm{v}$ is defined as $\text{supp}(\bm{v})\triangleq \{i \in [n]: v_{i} \neq 0\}$. The Hamming weight $wt(\bm{v})$ of $\bm{v}$ is defined as the size of $\text{supp}(\bm{v})$, i.e., $wt(\bm{v})=|\text{supp}(\bm{v})|$. The Hamming distance between $\bm u$ and $\bm v$ is defined as $d(\bm u, \bm v) \triangleq wt(\bm u-\bm v)$. A $q$-ary linear code $C$ of length $n$ and dimension $k$ is just a $k$-dimensional subspace of $\mathbb{F}^n_q$. The minimum distance of $C$ is defined as $d \triangleq \min\limits_{\bm c \in C, \bm c \neq \bm 0}wt(\bm c).$ An $[n, k, d]_q$-linear code $C$ is a linear code of length $n$, dimension $k$ and minimum distance $d$ over $\mathbb{F}_q$.

We recall some basic results about constacyclic codes, which are generalizations of cyclic codes. The readers may refer some standard textbooks on coding theory for more details on cyclic codes (\cite{HP03,MS77}) and some references on constacyclic codes (\cite{ASR01,B08,CFLL12,D10,FWF17,RZ08,YC15}).
\begin{definition}[$\lambda$-constacyclic codes and cyclic codes]\label{def1}
Suppose $\lambda \in \mathbb{F}^{*}_q$. A $q$-ary linear code $C$ of length $n$ is called a $\lambda$-constacyclic code if  $(c_0, c_1, \cdots, c_{n-1}) \in C$ implies that $(\lambda c_{n-1}, c_0, \cdots, c_{n-2}) \in C$. In particular, if $\lambda=1$, a $\lambda$-constacyclic code $C$ is call a cyclic code.
\end{definition}

If we identify a vector $\bm c=(c_0, c_1, \cdots, c_{n-1})$ with a polynomial $c(x)=c_0+c_1x+\cdots+c_{n-1}x^{n-1}$, then a $q$-ary $\lambda$-constacyclic code $C$ of length $n$ can be
identified with an ideal of the ring $\mathbb{F}_q[x]/(x^n-\lambda)$. Note that $\mathbb{F}_q[x]/(x^n-\lambda)$ is a principal domain, thus every ideal of $\mathbb{F}_q[x]/(x^n-\lambda)$ is generated by a monic polynomial $g(x)$ with $g(x) \mid (x^n-\lambda)$. We also define the Hamming weight  of $c(x)$ as $wt(c(x))\triangleq wt(\bm c)$, the Hamming weight of the corresponding vector $\bm c$.

\begin{definition}[Generator polynomial]
Suppose  $\lambda \in \mathbb{F}_q^*$ and $C$ is a $\lambda$-constacyclic code of length $n$  generated by a monic polynomial $g(x)$ with $g(x) \mid (x^n-\lambda)$, then $g(x)$ is called the generator polynomial of $C$, and we denote $C=\langle g(x) \rangle$. It is well-known that $\dim(C)=n-\deg(g(x)).$
\end{definition}
For any polynomial $f(x)=f_0+f_1x+\cdots+f_mx^m \in \mathbb{F}_q[x]$, where $f_0, f_m\neq0$, define its \emph{reciprocal polynomial} $\widetilde{f}(x)\triangleq f_m+f_{m-1}x+\cdots+f_0x^m=x^mf(x^{-1})$.
For convenience, we denote $\overline{f}(x)=f^{-1}_0\widetilde{f}(x)$, which is a monic polynomial.

The following lemma can be easily derived from the definitions.
\begin{lemma}\label{lem1}
For any polynomial $f(x), f_1(x) \in \mathbb{F}_q[x]$ with $f(0)$ and $f_1(0)$ are nonzeros, we have 
\begin{description}
\item[(i)] $\deg(f(x))=\deg(\overline{f}(x))$ and $wt(f(x))=wt(\overline{f}(x))$.
\item[(ii)] $f_1(x) \mid f(x)$ if and only if $\overline{f_1}(x) \mid \overline{f}(x)$.
\end{description}
\end{lemma}

Similar to cyclic codes, the dual code of  a $\lambda$-constacyclic code is a $\lambda^{-1}$-constacyclic code. Precisely, 

\begin{lemma}[\cite{D10,FWF17,YC15}]\label{lem2}
 Let $\lambda \in \mathbb{F}^*_{q}$ and $C=\langle g(x) \rangle$ be a $\lambda$-constacyclic code of length $n$, where $g(x)$ is the generator polynomial of $C$. Denote $h(x)=\frac{x^n-\lambda}{g(x)}$. Then the dual code $C^{\perp}$ of $C$ is a $\lambda^{-1}$-constacyclic code and $C^{\perp}=\langle \overline{h}(x)\rangle$.
\end{lemma}

 We now always assume that $\gcd(n,q)=1$.  Let $\mathbb{F}_{q^m}$ be the splitting field of $x^n-\lambda$ over $\mathbb{F}_q$. Suppose $\theta \in \mathbb{F}_{q^m}$ such that $\theta^n=\lambda$. Let $s$ be the order of $q$ modulo $n$, i.e., the least number $i$ such that $n \mid (q^i-1).$ Then we can choose a primitive $n$-th root of unity  $\alpha \in \mathbb{F}_{q^s}$ and the roots of $x^n-\lambda$ are $\theta, \theta \alpha, \cdots, \theta \alpha^{n-1}.$

The BCH bound and Hartmann-Tzeng bound \cite{HT72} are two well-known lower bounds on the minimum distance of  cyclic codes. In \cite{ASR01}, the authors presented the BCH bound for constacyclic codes. The Hartmann-Tzeng bound for constacyclic codes has been obtained in \cite{RZ08}. We give an alternative proof of the Hartmann-Tzeng bound for constacyclic codes for completeness. 

\begin{lemma}[BCH Bound for Constacyclic Codes \cite{ASR01}]\label{BCH}
 Let $C=\langle g(x) \rangle$ be a $\lambda$-constacyclic code of length $n$, where $g(x) \mid (x^n-\lambda)$ is the generator polynomial of $C$. Suppose $u$ and $b$ are integers with $\gcd(b,n)=1$. If  $g(\theta \alpha^{u+ib})=0$, for $i=0, 1, \cdots, \delta-2$, where $\delta \geq 2$, then the minimum distance $d(C)$ of $C$ is at least $\delta$.
\end{lemma}

\begin{lemma}[Hartmann-Tzeng Bound for Constacyclic Codes \cite{RZ08}]\label{HT}
 Let $C=\langle g(x) \rangle$ be a $\lambda$-constacyclic code of length $n$, where $g(x) \mid (x^n-\lambda)$ is the generator polynomial of $C$. Suppose $u, b_1$ and $b_2$ are integers with $\gcd(b_1,n)=\gcd(b_2,n)=1$. If  $g(\theta \alpha^{u+i_1b_1+i_2b_2})=0$, for $i_1=0, 1, \cdots, \delta-2$; $i_2=0,1,\cdots, \gamma$, where $\delta \geq 2$ and $\gamma \geq 0$, then the minimum distance $d(C)$ of $C$ is at least $\delta+\gamma$.
\end{lemma}

\begin{proof}
Firstly, by the conditions, we know the minimum distance $d(C) \geq \delta$ from Lemma \ref{BCH}. Let $c(x) \in C$ be a nonzero codeword with $wt(c(x))=w$. We suppose $\delta \leq w \leq \delta+\gamma-1.$ Since $C$ is a $\lambda$-constacyclic code, without loss of generality, we may assume that 
\[c(x)=1+\sum_{j=1}^{w-1}c_jx^{k_j},\]
where $c_i \in \mathbb{F}^{*}_q$ and $0 < k_1<k_2< \cdots<k_{w-1}< n$. Denote 
\[S_i=c(\theta \alpha^i)-1=\sum_{j=1}^{w-1}c_j(\theta \alpha^i)^{k_j}.\]
Note that $g(x) \mid c(x)$, hence $c(\theta \alpha^{u+i_1b_1+i_2b_2})=0$, for $i_1=0, 1, \cdots, \delta-2$; $i_2=0,1,\cdots, \gamma$. Thus 
\[S_{u+i_1b_1+i_2b_2}=-1, \textnormal{ for } i_1=0, 1, \cdots, \delta-2; i_2=0,1,\cdots, \gamma.\]
Let 
\[f_1(x)=\prod_{j_1=1}^{\delta-2}(x-\alpha^{k_{j_1}b_1})=\sum_{i_1=0}^{\delta-2}\sigma_{i_1}x^{i_1},\]
and
\[f_2(x)=\prod_{j_2=\delta-1}^{w-1}(x-\alpha^{k_{j_2}b_2})=\sum_{i_2=0}^{w-\delta-1}\tau_{i_2}x^{i_2}.\]
Since $\gcd(b_1,n)=\gcd(b_2,n)=1$ and $0 < k_1<k_2< \cdots<k_{w-1}< n$, we have $\alpha^{k_{j_1}b_1}\neq 1$ and $\alpha^{k_{j_2}b_2}\neq 1$ for any $1 \leq j_1, j_2 \leq w-1$. Thus
\[f_1(1)f_2(1)=\prod_{j_1=1}^{\delta-2}(1-\alpha^{k_{j_1}b_1})\prod_{j_2=\delta-1}^{w-1}(1-\alpha^{k_{j_2}b_2}) \neq 0.\]
On the other hand,
\begin{eqnarray*}
f_1(1)f_2(1)&=& \sum_{i_2=0}^{w-\delta-1}\tau_{i_2}\sum_{i_1=0}^{\delta-2}\sigma_{i_1}= -\sum_{i_2=0}^{w-\delta-1}\tau_{i_2}\sum_{i_1=0}^{\delta-2}S_{u+i_1b_1+i_2b_2}\sigma_{i_1} \\
&=&-\sum_{i_2=0}^{w-\delta-1}\tau_{i_2}\sum_{i_1=0}^{\delta-2}\sigma_{i_1}\sum_{j=1}^{w-1}c_{j}(\theta \alpha^{u+i_1b_1+i_2b_2})^{k_j} \\
&=& -\sum_{i_2=0}^{w-\delta-1}\tau_{i_2}\sum_{j=1}^{w-1}c_{j}(\theta\alpha^{u+i_2b_2})^{k_j}\sum_{i_1=0}^{\delta-2}\sigma_{i_1}(\alpha^{k_jb_1})^{i_1} \\
&=&-\sum_{i_2=0}^{w-\delta-1}\tau_{i_2}\sum_{j=1}^{w-1}c_{j}(\theta\alpha^{u+i_2b_2})^{k_j}f_1(\alpha^{k_jb_1})\\
&=&-\sum_{i_2=0}^{w-\delta-1}\tau_{i_2}\sum_{j=\delta-1}^{w-1}c_{j}(\theta\alpha^{u+i_2b_2})^{k_j}f_1(\alpha^{k_jb_1})\\
&=&-\sum_{j=\delta-1}^{w-1}c_{j}\theta^{k_j}\alpha^{uk_j}f_1(\alpha^{k_jb_1})\sum_{i_2=0}^{w-\delta-1}\tau_{i_2}(\alpha^{k_jb_2})^{i_2}\\
&=&-\sum_{j=\delta-1}^{w-1}c_{j}\theta^{k_j}\alpha^{uk_j}f_1(\alpha^{k_jb_1})f_2(\alpha^{k_jb_2})=0,
\end{eqnarray*}
where the 6-th equality holds since $f_1(\alpha^{k_jb_1})=0$ for any $1 \leq j \leq \delta-2$, and the last equality holds since $f_2(\alpha^{k_jb_2})=0$ for any $\delta-1 \leq j \leq w-1$. Thus it leads to a contradiction. So we have $w \geq \delta+\gamma,$ hence $d(C)\geq \delta+\gamma$. The lemma is proved.
\end{proof}

Recall that $C$ is called an $r$-LRC if and only if each symbol $c_i$ of any codeword $\bm c$ can be represented as a linear combination of $\{c_j\}_{j\in R_i\setminus \{i\}}$ for some $R_i\subset [n]$ with $i \in R_i$ and $|R_{i}| \leq r+1$.  It can be equivalently defined in a viewpoint of dual codes:
\begin{proposition}[\cite{HXSCFY20}]\label{parity}
Let $C$ be a linear code of length $n$. Then the $i$-th symbol of $C$ has locality $r$ if and only if there exists a dual codeword $\bm h\in C^{\perp}$  such that
\[i \in \textnormal{supp}(\bm h) \textnormal{ and } |\textnormal{supp}(\bm h)| \leq r+1.\]
\end{proposition}

In this paper, an $(n,k,d;r)$-LRC means a linear code of length $n$, dimension $k$, minimum distance $d$ and locality $r$.

Due to their nice algebraic structures, cyclic and constacyclic codes have been widely used in the constructions of optimal LRCs (or $(r, \delta)$-LRCs), e.g.,  \cite{CFXF19,CXHF18,TB14}. The following proposition presented a sufficient condition to ensure that a constacyclic code is an LRC.
\begin{proposition}[\cite{CFXF19}]\label{locality}
Suppose $(r+1) \mid n$, $\gcd(n,q)=1$, $\lambda=\theta^n \in \mathbb{F}_q^{*}$ and $\alpha$ is a primitive $n$-th root of unity. Let $C=\langle g(x)\rangle$ be a constacyclic code of length $n$, where $g(x) \mid (x^n-\lambda)$ is the generator polynomial of $C$. If  $g(\theta\alpha^{j}\alpha^{i(r+1)})=0$, for $i=0,1,\dots, \frac{n}{r+1}-1$ and some integer $j$, then $C$ has locality $r$.
\end{proposition}

\begin{proof}
Since $g(\theta\alpha^{j}\alpha^{i(r+1)})=0$, for $i=0,1,\dots, \frac{n}{r+1}-1$, we have \[\prod_{i=0}^{\frac{n}{r+1}-1}(x-\theta\alpha^{j}\alpha^{i(r+1)}) \mid g(x),\]
i.e., $(x^{\frac{n}{r+1}}-(\theta\alpha^j)^{\frac{n}{r+1}}) \mid g(x)$. Denote $\xi =\theta\alpha^j$, then $\xi^n=\lambda$ and  $x^n-\lambda=x^n-\xi^n=(x^{\frac{n}{r+1}}-\xi^{\frac{n}{r+1}})f(x),$ where $f(x)=\sum_{i=0}^{r}x^{i\frac{n}{r+1}}\xi^{(r-i)\frac{n}{r+1}}.$ Then $wt(f(x))=r+1$ and 
\[h(x)=\frac{x^n-\lambda}{g(x)} \mid  \frac{x^n-\lambda}{x^{\frac{n}{r+1}}-\xi^{\frac{n}{r+1}}}=f(x).\]
From Lemma \ref{lem1}, we have $\overline{h}(x) \mid \overline{f}(x)$. Note that $C^{\perp}=\langle\overline{h}(x)\rangle$ by Lemma \ref{lem2} and $wt(\overline{f}(x))=wt(f(x))=r+1$ by Lemma \ref{lem1}. Thus $\overline{f}(x)$ is a codeword of $C^{\perp}$ of weight $r+1$. According to Proposition \ref{parity}, the first symbol has $r$-locality. Then we can deduce that all symbols of $C$ have $r$-locality since the constacyclicity of $C$.
\end{proof}
\section{Cyclic Singleton-Optimal LRCs}

\subsection{Cyclic Singleton-Optimal LRCs for even $q$}\label{III-A}
In this subsection, for $q$ is even and $3 \mid (q-1)$, we will give a construction of $q$-ary cyclic Singleton-optimal LRCs with length $n=3(q+1)$, minimum distance $d=6$ and locality $r=2$. Note that $n \mid (q^2-1)$, thus we can choose $\alpha \in \mathbb{F}_{q^2}$, which is a primitive $n$-th root of unity.

\begin{theorem}\label{even}
With the above notation, let
\[L=\{1, \alpha^{3}, \alpha^6, \cdots, \alpha^{3q}\}\]
and
\[D=\{\alpha, \alpha^{q-1}, \alpha^{q}, \alpha^{q+1}, \alpha^{q+2}\}.\]
Set
\[g(x)=\prod\limits_{a \in L \bigcup D}(x-a)\]
and $C=\langle g(x)\rangle$. Then  $C$  is a $q$-ary cyclic Singleton-optimal LRC with length $n=3(q+1)$, minimum distance $d=6$ and locality $r=2$.
\end{theorem}

\begin{proof}
Since $\alpha^{3(q+1)}=1$, we have  
\[\prod_{a\in L}(x-a)=\prod_{i=0}^{q}(x-\alpha^{3i})=x^{q+1}-1.\] 
Since $3 \mid (q-1)$, we have $L\bigcap D= \{\alpha^{q-1}, \alpha^{q+2}\}$, thus \[g(x)=(x^{q+1}-1)(x-\alpha)(x-\alpha^{q})(x-\alpha^{q+1}).\]
Since $\alpha^{q+1} \in \mathbb{F}_{q}$, thus
$g(x) \in \mathbb{F}_{q}[x]$ and $k=\dim(C)=n-\deg(g(x))=2q-1$. According to the set $L$,  $C$ has locality $r=2$ by Proposition \ref{locality}. Since $\alpha^{q-1}, \alpha^{q}, \alpha^{q+1}, \alpha^{q+2}$ are roots of $g(x)$, we have $d \geq 5$ by Lemma \ref{BCH}. On the other hand, by the Singleton-type bound \eqref{singleton}, we have
\[d \leq n-k+2-\left\lceil k/r\right\rceil=3(q+1)-(2q-1)+2-q=6.\]
Thus we only need to prove that $d \neq 5$. By contradiction, we suppose $d=5$. Let \[c(x)=c_0+c_1x^{k_1}+c_2x^{k_2}+c_3x^{k_3}+c_4x^{k_4}\]
be a codeword of $C$ with $c_0, c_1, c_2, c_3, c_4 \in \mathbb{F}_{q}^{*}$ and $k_0=0 < k_1, k_2,  k_3, k_4 \leq n-1$, where $k_1, k_2,  k_3, k_4$ are pairwise distinct. Then $g(x) \mid c(x)$.  Let
\[I=\{\ell : 0 \leq \ell \leq 4 \textnormal{ and }(q+1) \mid k_{\ell} \}.\] Since $k_0=0<k_1, k_2,  k_3, k_4 \leq n-1=3(q+1)-1,$ we have  $1\leq |I| \leq 3$. Note that $\alpha$ is a primitive $n$-th root of unity, we can show that
\begin{equation}\label{5}
            \sum_{i=0}^{q}(\alpha^{3i})^{k_{\ell}}=\left\{
            \begin{aligned}
                & 1, ~\textnormal{ for } \ell \in I, \\
                & 0 ,~ \textnormal{ for } \ell \notin I.
            \end{aligned}\right.
        \end{equation}
For any $0 \leq i \leq q,$ we have $c(\alpha^{3i})=0$. Thus by Eq. \eqref{5}
\begin{eqnarray*}
  0 &=& \sum_{i=0}^{q}c(\alpha^{3i})=\sum_{i=0}^{q}\sum_{\ell=0}^{4}c_{\ell}(\alpha^{3i})^{k_{\ell}} \\
    &=& \sum_{\ell=0}^{4}c_{\ell} \sum_{i=0}^{q}(\alpha^{3i})^{k_{\ell}}=\sum_{\ell \in I}c_{\ell},
\end{eqnarray*}
which implies that $|I| \geq 2$.

\textbf{Case 1}: $|I|=3$. Without loss of generality, we suppose \[I=\{0,1,2\},  k_1=q+1, k_2=2(q+1), \textnormal{ and } (q+1) \nmid k_3 .\]
Thus $c_0+c_1+c_2=0$. By $c(1)=0$, we have \[c_0+c_1+c_2+c_3+c_4=0 \Rightarrow c_3+c_4=0.\]
By $c(\alpha^3)=0$, we have
\[c_0+c_1(\alpha^3)^{q+1}+c_2(\alpha^3)^{2(q+1)}+c_3(\alpha^3)^{k_3}+c_4(\alpha^3)^{k_4}=0\]
\[\Rightarrow c_3\alpha^{3k_3}+c_4\alpha^{3k_4}=0 \Rightarrow \alpha^{3k_3}=\alpha^{3k_4}\]
\[\Rightarrow (q+1) \mid (k_4-k_3).\]
By $c(\alpha)=0$, we have \[c_0+c_1\alpha^{q+1}+c_2\alpha^{2(q+1)}+c_3\alpha^{k_3}+c_4\alpha^{k_4}=0 \]
\[\Rightarrow c_3\alpha^{k_3}+c_4\alpha^{k_4} \in \mathbb{F}_{q} \Rightarrow c_3\alpha^{k_3}(1-\alpha^{k_4-k_3}) \in \mathbb{F}_{q}.\]
Since $(q+1) \mid (k_4-k_3)$ and $n \nmid (k_4-k_3)$,  we have $\alpha^{k_4-k_3} \neq 1 \in \mathbb{F}_{q}$. Thus \[\alpha^{k_3}\in \mathbb{F}^{*}_{q} \Rightarrow \alpha^{k_3(q-1)}=1 \Rightarrow n \mid k_3(q-1).\]
Since $q$ is even, $\gcd(\frac{n}{3},q-1)=\gcd(q+1, q-1)=1$. Thus $\frac{n}{3} \mid k_3$, which is a contradiction.

\textbf{Case 2}: $|I|=2$. Without loss of generality, we suppose \[I=\{0,1\} \textnormal{ and } (q+1) \mid k_1, (q+1) \nmid k_2, k_3,k_4.\]
Thus $c_0+c_1=0$. Similarly, from $c(1)=0$, we have $c_2+c_3+c_4=0$. From $c(\alpha^3)=c(\alpha^6)=0$, we deduce that
\[c_2\alpha^{3k_2}+c_3\alpha^{3k_3}+c_4\alpha^{3k_4}=0,\]
\[ c_2\alpha^{6k_2}+c_3\alpha^{6k_3}+c_4\alpha^{6k_4}=0.\]
Thus
\[\det\left(
        \begin{array}{ccc}
          1 & 1 & 1 \\
          \alpha^{3k_2} & \alpha^{3k_3} & \alpha^{3k_4} \\
          \alpha^{6k_2} & \alpha^{6k_3} & \alpha^{6k_4} \\
        \end{array}
      \right)=0,\]
which implies that $\alpha^{3k_2}=\alpha^{3k_3}$, or $\alpha^{3k_2}=\alpha^{3k_4}$, or $\alpha^{3k_3}=\alpha^{3k_4}$. Without loss of generality, we assume that $\alpha^{3k_2}=\alpha^{3k_3}$. We next prove that $\alpha^{3k_2}=\alpha^{3k_3}=\alpha^{3k_4}$, otherwise, suppose that $\alpha^{3k_2}=\alpha^{3k_3} \neq \alpha^{3k_4}$. Then for $\ell=1,2,3,$ \[\sum_{i=0}^{q}\alpha^{3i(k_{\ell}-k_4)}=0.\]
From $c(\alpha^{3i})=0$ $(i=0,1,2, \cdots, q)$, we have
\begin{eqnarray*}
  0 &=& \sum_{i=0}^{q}\alpha^{-3ik_4}c(\alpha^{3i}) \\
    &=& \sum_{i=0}^{q}\alpha^{-3ik_4}\sum_{\ell=0}^{4}c_{\ell}\alpha^{3ik_{\ell}}  \\
    &=&  \sum_{\ell=0}^{4}c_{\ell}\sum_{i=0}^{q}\alpha^{3i(k_{\ell}-k_4)} \\
    &=& c_4,
\end{eqnarray*}
which is a contradiction. So $\alpha^{3k_2}=\alpha^{3k_3}=\alpha^{3k_4} \Rightarrow (q+1) \mid (k_2-k_4)$ and $(q+1) \mid (k_3-k_4)$. From $c(\alpha)=0$, we have \[c_0+c_1\alpha^{k_1}+c_2\alpha^{k_2}+c_3\alpha^{k_3}+c_4\alpha^{k_4}=0.\]
Since $\alpha \in \mathbb{F}_{q^2}$ and $(q+1) \mid k_1$, we have $c_0+c_1\alpha^{k_1} \in \mathbb{F}_{q}$. Since $c_0+c_1=0,$  $c_0+c_1\alpha^{k_1}=c_0(1-\alpha^{k_1}) \in \mathbb{F}^{*}_{q}$. Thus \[c_2\alpha^{k_2}+c_3\alpha^{k_3}+c_4\alpha^{k_4} \in \mathbb{F}^{*}_{q},\]
i.e., 
\[c_4\alpha^{k_4}(1+\frac{c_3}{c_4}\alpha^{k_3-k_4}+\frac{c_2}{c_4}\alpha^{k_2-k_4}) \in \mathbb{F}^{*}_{q}.\]
Since $(q+1) \mid (k_2-k_4)$ and $(q+1) \mid (k_3-k_4)$, we have $\alpha^{k_3-k_4}, \alpha^{k_2-k_4} \in \mathbb{F}_{q}$, hence $\alpha^{k_4} \in \mathbb{F}^{*}_{q}$. Then $n \mid k_4(q-1)$. Note that $\gcd(\frac{n}{3}, q-1)=\gcd(q+1,q-1)=1$, thus $(q+1) \mid k_4$, which is a contradiction.

In summary, $d=6$ and hence $C$ is a $q$-ary cyclic Singleton-optimal LRC with length $n=3(q+1)$, minimum distance $d=6$ and locality $r=2$.
\end{proof}

\begin{remark}
To the best of our knowledge, Theorem \ref{even} is the first construction of $q$-ary cyclic Singleton-optimal LRCs with length $n>q+1$ and minimum distance  $d \geq 5$.
\end{remark}

\begin{example}
Suppose $q=4$ and $n=3(q+1)=15$. Let $\omega$ be a primitive element of $\mathbb{F}_{4}$ and $\alpha \in \mathbb{F}_{16}$ which is a root of the primitive polynomial $x^2+x+\omega \in \mathbb{F}_4[x]$. Then $\alpha+\alpha^4=1$ and $\alpha^5=\omega$. Let
\begin{eqnarray*}
g(x) &=& (x^5-1)(x-\alpha)(x-\alpha^4)(x-\alpha^5) \\
  &=& (x^5-1)(x^2-x+\omega)(x-\omega) \in \mathbb{F}_{4}[x]
\end{eqnarray*}
and $C$ be the cyclic code of length 15 with the generator polynomial $g(x)$. With the help of Magma, we can calculate that $d=6$, which coincides with Theorem \ref{even}.
\end{example}

The condition of $q$ is even is necessary in Theorem \ref{even}. Indeed,  we show that when $q$ is odd, the minimum distance of the LRCs constructed in Theorem \ref{even} is equal to  5.

\begin{proposition}\label{prop3}
Suppose $q$ is odd and $3 \mid (q-1)$. Let $C$ be the cyclic code of length $n=3(q+1)$ constructed as in Theorem \ref{even}. Then the minimum distance of $C$ is 5.
\end{proposition}

\begin{proof}
From the proof of Theorem \ref{even}, the generator polynomial of $C$ is \[g(x)=(x^{q+1}-1)(x-\alpha)(x-\alpha^{q})(x-\alpha^{q+1}),\]
where $\alpha \in \mathbb{F}_{q^{2}}$ is a primitive $n$-th root of unity. Denote $\pi=\alpha^{\frac{q+1}{2}}\in \mathbb{F}_q$. Then $\pi^3=-1$ and $\pi^2-\pi+1=0$. Let
\[c(x)=1-2x^{q+1}+x^{2(q+1)}-\frac{3}{2\pi-1}x^{\frac{q+1}{2}}+\frac{3}{2\pi-1}x^{\frac{3(q+1)}{2}} \in \mathbb{F}_q[x].\] Then $wt(c(x))=5$. From $\alpha$ is a primitive $n$-th root of unity, for $i=0,1,\dots,q$, we have
\[c(\alpha^{3i})=1-2+1-\frac{3}{2\pi-1}(-1)^i+\frac{3}{2\pi-1}(-1)^{3i}=0.\]
Moreover, 
\begin{eqnarray*}
c(\alpha)&=&1-2\pi^2+\pi^4-\frac{3}{2\pi-1}\pi+\frac{3}{2\pi-1}\pi^3\\
&=&1-2(\pi-1)-\pi-\frac{3(\pi+1)}{2\pi-1} \\
&=& \frac{3((1-\pi)(2\pi-1)-(\pi+1))}{2\pi-1} \\
&=& \frac{-6(\pi^2-\pi+1)}{2\pi-1}=0,
\end{eqnarray*}
\[c(\alpha^q)=(c(\alpha))^q=0,\]
and from $3 \mid (q-1)$, we have
\begin{eqnarray*}
c(\alpha^{q+1})=c(\pi^2)&=&1-2\pi^{2(q+1)}+\pi^{4(q+1)}-\frac{3}{2\pi-1}\pi^{q+1}+\frac{3}{2\pi-1}\pi^{3(q+1)}\\
&=&1+2\pi+\pi^2-\frac{3}{2\pi-1}\pi^2+ \frac{3}{2\pi-1}\\
&=& 3\pi-\frac{3\pi-6}{2\pi-1}=\frac{3\pi(2\pi-1)-(3\pi-6)}{2\pi-1} \\
&=& \frac{6(\pi^2-\pi+1)}{2\pi-1}=0.
\end{eqnarray*}

Thus $g(x) \mid c(x)$ and hence $c(x)$ is a nonzero codeword of $C$, which implies that  $d \leq 5$. Since $\alpha^{q-1}, \alpha^{q}, \alpha^{q+1}, \alpha^{q+2}$ are roots of $g(x)$, by the BCH bound, $d \geq 5$. Thus $d=5$. The proof is completed.
\end{proof}

\begin{example}
Suppose $q=7$ and $n=3(q+1)=24$. Let $\alpha \in \mathbb{F}_{49}$ be a primitive 24-th root of unity. Let $g(x) = (x^8-1)(x-\alpha)(x-\alpha^7)(x-\alpha^8)=x^{11} + 3x^{10} + 6x^9 + 3x^8 + 6x^3 + 4x^2 + x + 4  \in \mathbb{F}_{7}[x]$
and $C$ be the cyclic code of length 24 with generator polynomial $g(x)$. With the help of Magma, we can calculate that $d=5$, which coincides with Proposition \ref{prop3}.
\end{example}

\subsection{Cyclic Singleton-optimal LRCs for $q\equiv 1 (\bmod~4)$}
In Subsection \ref{III-A}, we give a construction of $q$-ary cyclic Singleton-optimal LRCs of length $n=3(q+1)$ for $q$ is even. And we also have proved that this construction is not optimal for $q$ is odd in Proposition \ref{prop3}. In this subsection, we construct a family $q$-ary Singleton-optimal LRCs of length $n=\frac{3(q+1)}{2}$ for $q\equiv 1 (\bmod~4)$ and $3 \mid (q-1)$.

\begin{theorem}\label{odd}
Suppose $q\equiv 1 (\bmod~4)$ and $3 \mid (q-1)$. Let $n=\frac{3(q+1)}{2}$ and $\alpha \in \mathbb{F}_{q^{2}}$ be a primitive $n$-th root of unity. Let
\[g(x)=(x^{\frac{n}{3}}-1)(x-\alpha)(x-\alpha^{q})(x-\alpha^{q+1}).\] Then the cyclic code $C\triangleq\langle g(x)\rangle$  is a $q$-ary cyclic Singleton-optimal LRC with length $n$, minimum distance $d=6$ and locality $r=2$.
\end{theorem}

\begin{proof}
From $(x^{\frac{n}{3}}-1) \mid g(x)$, we know that $1, \alpha^{3}, \alpha^{6}, \cdots, \alpha^{n-3}$ are roots of $g(x)$. By Proposition \ref{locality}, $C$ has locality $r=2$. Since $3 \mid (q-1)$ and $\alpha^{q-1}, \alpha^{q}, \alpha^{q+1}, \alpha^{q+2}$ are roots of $g(x)$, we have $d \geq 5$ by Lemma \ref{BCH}. The dimension $k$ of $C$ is equal to $n-\deg(g(x))=q-2$. By the Singleton-type bound, 
\[d \leq \frac{3(q+1)}{2}-(q-2)+2-\lceil\frac{q-2}{2}\rceil=6.\] Thus we only need to prove that $d \neq 5$. By contradiction, let \[c(x)=c_0x^{k_0}+c_1x^{k_1}+c_2x^{k_2}+c_3x^{k_3}+c_4x^{k_4}\]
be a codeword of $C$ with $c_0, c_1, c_2, c_3, c_4 \in \mathbb{F}_{q}^{*}$ and $0=k_0< k_1< k_2<  k_3< k_4 \leq n-1$. Let
\[I=\{\ell: 0 \leq \ell \leq 4 \textnormal{ and }\frac{n}{3} \mid  k_{\ell}\}.\]
Then $1\leq |I| \leq 3$. Similar to the proof of Theorem \ref{even}, we have 
\[\sum_{\ell \in I}c_{\ell}=0,\]
and hence $|I| \geq 2$.

\textbf{Case 1}: $|I|=3$. Without loss of generality, we suppose \[I=\{0,1,2\}, k_1=\frac{n}{3}, k_2=\frac{2n}{3} \textnormal{ and }\frac{n}{3} \nmid k_3, k_4.\]
Thus $c_0+c_1+c_2=0$. By $c(1)=0$, we have \[c_0+c_1+c_2+c_3+c_4=0 \Rightarrow c_3+c_4=0.\]
By $c(\alpha^3)=0$, we have
\[c_0+c_1(\alpha^3)^{\frac{n}{3}}+c_2(\alpha^3)^{\frac{2n}{3}}+c_3(\alpha^3)^{k_3}+c_4(\alpha^3)^{k_4}=0,\]
\[\Rightarrow c_3\alpha^{3k_3}+c_4\alpha^{3k_4}=0 \Rightarrow \alpha^{3k_3}=\alpha^{3k_4}\Rightarrow \frac{n}{3} \mid (k_4-k_3).\]
By $c(\alpha)=0$, we have \[c_0+c_1\alpha^{\frac{n}{3}}+c_2\alpha^{\frac{2n}{3}}+c_3\alpha^{k_3}+c_4\alpha^{k_4}=0. \]
Note that $(\alpha^{\frac{n}{3}})^{q-1}=1$, thus $\alpha^{\frac{n}{3}} \in \mathbb{F}_{q}$.
\[\Rightarrow c_3\alpha^{k_3}+c_4\alpha^{k_4} \in \mathbb{F}_{q} \Rightarrow c_3\alpha^{k_3}(1-\alpha^{k_4-k_3}) \in \mathbb{F}_{q}.\]
Since $\frac{n}{3} \mid (k_4-k_3)$ and $n \nmid (k_4-k_3)$,  we have $\alpha^{k_4-k_3} \neq 1$ and $\alpha^{k_4-k_3}-1 \in \mathbb{F}_{q}$. Thus 
\[\alpha^{k_3}\in \mathbb{F}_{q} \Rightarrow n \mid k_3(q-1).\]
Note that $\gcd(\frac{n}{3}, q-1)=\gcd(\frac{q+1}{2}, q-1)=1$ by $q \equiv 1 (\bmod~4)$. Thus $\frac{n}{3} \mid k_3$, which is a contradiction.

\textbf{Case 2}: $|I|=2$. The proof of this case is completely similar to the above case and the proof of Theorem \ref{even}, so we omit it.

The theorem is proved.

\end{proof}
\begin{remark}
Suppose $q$ is odd and $3 \mid (q-1)$. If $q+1=2^a\cdot b$, where $b$ is an odd, then let $n=\frac{3(q+1)}{2^a}=3b$ and $\alpha \in \mathbb{F}_{q^{2}}$ be a primitive $n$-th root of unity. Let
\[g(x)=(x^{\frac{n}{3}}-1)(x-\alpha)(x-\alpha^{q})(x-\alpha^{q+1}).\]
Similar to the proof of Theorem \ref{odd},  we can show that $C\triangleq\langle g(x)\rangle$  is a $q$-ary cyclic Singleton-optimal LRC with length $n$, minimum distance $d=6$ and locality $r=2$.
\end{remark}
\begin{example}
Suppose $q=13$ and $n=3(q+1)/2=21$. Let $\alpha \in \mathbb{F}_{13^2}$ be a primitive 21-th root of unity. Let $g(x) = (x^7-1)(x-\alpha)(x-\alpha^{13})(x-\alpha^{14})=x^{10} + 9x^9 + 3x^8 + 10x^7 + 12x^3 + 4x^2 + 10x + 3  \in \mathbb{F}_{13}[x]$
and $C$ be the cyclic code of length 21 with generator polynomial $g(x)$. With the help of Magma, we can calculate that $d=6$, which coincides with Proposition \ref{thm3}.
\end{example}

\section{Cyclic and Constacyclic Perfect LRCs}
In \cite{FCXF20-2}, we first introduced the definition of perfec LRCs. By using the techniques of finite geometry and finite fields, we proposed two general constructions of perfect LRCs in \cite{FCXF20-2}. In this section, we will give two constructions of perfect LRCs via cyclic and constacyclic codes.
\subsection{Cyclic Perfect LRCs}\label{cyclic-perfect}
 Before give our construction of cyclic perfect LRCs, we need an auxiliary lemma.
\begin{lemma}[\cite{LN97}]\label{Tr}
Let $a \in \mathbb{F}_{q}$  and $p$ be the characteristic of $\mathbb{F}_{q}$. Then the trinomial $x^{p}-x-a$ is irreducible in $\mathbb{F}_{q}[x]$ if and only if $\textnormal{Tr}_{\mathbb{F}_{q}/\mathbb{F}_{p}}(a) \neq 0$.
\end{lemma}

Suppose $q \geq 4$ is even with $3 \mid (q+1)$, and $m$ is even with $\gcd(m,q-1)=1$.
Let $n=\frac{q^{m}-1}{q-1}$, then $n=q^{m-1}+q^{m-2}+\cdots +1 \equiv 0 (\textnormal{mod } q+1)$. Now, our construction of cyclic perfect LRCs is presented as follows.
\begin{theorem}\label{thm3}
Suppose $q \geq 4$ is even with $3 \mid (q+1)$, and $m$ is even with $\gcd(m,q-1)=1$. Let $\alpha \in \mathbb{F}_{q^{m}}$ be a primitive $n$-th root of unity. Set
\[L=\{1, \alpha^{3}, \alpha^{6}, \cdots, \alpha^{n-3}\}\]
and
\[D=\{\alpha, \alpha^{q}, \cdots, \alpha^{q^{m-1}}\}.\]
Let $g(x)=\prod\limits_{a \in L\bigcup D}(x-a)$ and $C=\langle g(x)\rangle$.
Then $C$ is a $q$-ary cyclic $(n=\frac{q^{m}-1}{q-1}, k=\frac{2}{3}n-m, d=5; r=2)$-perfect LRC.
\end{theorem}

\begin{proof}
The elements of $L$ and $D$ are roots of $x^n-1$, and
\[\prod_{a \in L}(x-a)=\prod_{i=0}^{\frac{n}{3}-1}(x-\alpha^{3i})=x^\frac{n}{3}-1.\]
Then
\[g(x)=\prod\limits_{a \in L\bigcup D}(x-a)=(x^{\frac{n}{3}}-1)(x-\alpha)(x-\alpha^q)\cdots(x-\alpha^{q^{m-1}}).\]
Thus $g(x) \in \mathbb{F}_q[x]$ and $g(x) \mid (x^n-1)$. So $C$ is a cyclic code over $\mathbb{F}_q$ of length $n$ and the dimension of $C$ is $k=n-\deg(g(x))=\frac{2}{3}n-m$.

Note that $\alpha^{0}, \alpha^{1}, \alpha^{q}, \alpha^{q+1}$ are roots of $g(x)$ and $\gcd(q,n)=1$. By the Hartmann-Tzeng bound (let $\lambda=1, \delta=3, \gamma=1, u=0, b_1=1, b_2=q$ in Lemma \ref{HT}), we have $d \geq 3+1=4$. Next we prove that $d \neq 4$. Otherwise, suppose there exists a codeword 
\[c(x)=c_0+c_1x^{k_1}+c_2x^{k_2}+c_3x^{k_3}\in C\]
with $c_0, c_1, c_2, c_3 \in \mathbb{F}_{q}^{*}$ and $0=k_0< k_1, k_2,  k_3 \leq n-1$, where $k_1, k_2,  k_3$ are pairwise distinct positive integers. Then $g(x) \mid c(x)$.  Let
\[I=\{\ell: 0 \leq \ell \leq 3 \textnormal{ and }  \frac{n}{3}\mid  k_{\ell} \}.\] Then $0 \in I$. Since $k_1, k_2,  k_3\leq n-1$, $|I| \leq 3$.
Note that
\begin{equation*}\label{sum}
            \sum_{i=0}^{\frac{n}{3}-1}(\alpha^{3i})^{k_{\ell}}=\left\{
            \begin{aligned}
                & 1, ~\textnormal{ for } \ell \in I, \\
                & 0 ,~ \textnormal{ for } \ell \notin I.
            \end{aligned}\right.
        \end{equation*}
From $c(\alpha^{3i})=0$, $i=0, 1, \cdots, \frac{n}{3}-1$, we have 
\begin{eqnarray*}
  0 &=& \sum_{i=0}^{\frac{n}{3}-1}c(\alpha^{3i})=\sum_{i=0}^{\frac{n}{3}-1}\sum_{\ell=0}^{3}c_{\ell}(\alpha^{3i})^{k_{\ell}} \\
    &=& \sum_{\ell=0}^{3}c_{\ell} \sum_{i=0}^{\frac{n}{3}-1}(\alpha^{3i})^{k_{\ell}}=\sum_{\ell \in I}c_{\ell},
\end{eqnarray*} 
which implies that $|I| \geq 2$.

\textbf{Case 1}: $|I|=3$. Without loss of generality, we suppose
\[I=\{0,1,2\}, k_1=\frac{n}{3}\textnormal{ and } k_2=\frac{2n}{3}.\]
Thus $c_0+c_1+c_2=0$. By $c(1)=0$, we have $c_0+c_1+c_2+c_3=0 \Rightarrow c_3=0$, which is a contradiction.

\textbf{Case 2}: $|I|=2$. Without loss of generality, we suppose
\[I=\{0,1\}, \frac{n}{3} \mid k_1\textnormal{ and }\frac{n}{3} \nmid k_2, k_3.\]
Thus $c_0+c_1=0$. Similarly, from $c(1)=0$, we have $c_2+c_3=0$. From $c(\alpha^{3})=0$, we have \[c_{0}+c_{1}+c_{2}\alpha^{3k_2}+c_{3}\alpha^{3k_3}=0 \]
\[\Rightarrow \alpha^{3k_2}=\alpha^{3k_3} \Rightarrow \frac{n}{3} \mid (k_3-k_2).\]
From $c(\alpha)=0$, we have $c_{0}+c_{1}\alpha^{k_1}+c_{2}\alpha^{k_2}+c_{3}\alpha^{k_3}=0 \Rightarrow c_{0}(1+\alpha^{k_1})+c_{2}\alpha^{k_2}(1+\alpha^{k_3-k_2})=0 \Rightarrow \frac{\alpha^{k_2}(1+\alpha^{k_3-k_2})}{(1+\alpha^{k_1})}=\frac{c_0}{c_2} \in \mathbb{F}_{q} \Rightarrow$
\[\Big(\frac{\alpha^{k_2}(1+\alpha^{k_3-k_2})}{(1+\alpha^{k_1})}\Big)^{q}=\frac{\alpha^{k_2}(1+\alpha^{k_3-k_2})}{(1+\alpha^{k_1})}.\]
Since $3 \mid (q+1)$ and $\frac{n}{3} \mid k_1, (k_3-k_2)$, we have $\alpha^{(q+1)k_1}=\alpha^{(q+1)(k_3-k_2)}=1.$ Thus \[\Big(\frac{\alpha^{k_2}(1+\alpha^{k_3-k_2})}{(1+\alpha^{k_1})}\Big)^{q}=\frac{\alpha^{qk_2}(1+\alpha^{-(k_3-k_2)})}{(1+\alpha^{-k_1})}=\frac{\alpha^{qk_2+k_1}(1+\alpha^{k_3-k_2})}{\alpha^{k_3-k_2}(1+\alpha^{k_1})}.\]
Then $\alpha^{k_2}=\frac{\alpha^{qk_2+k_1}}{\alpha^{k_3-k_2}} \Rightarrow n \mid \big((q-1)k_2+k_1+k_3-k_2\big) \Rightarrow \frac{n}{3} \mid (q-1)k_2$.
Note that $\gcd(n,q-1)=\gcd(m,q-1)=1$, thus $\frac{n}{3} \mid k_2$, which is a contradiction. Therefore $d \geq 5$.

According to the set $L$, $C$ has locality $r=2$ by Proposition \ref{locality}. Note that $\frac{q^{\frac{r}{r+1}n}}{\frac{rn}{2}(q-1)+1}=\frac{q^\frac{2}{3}n}{q^m}=q^k$, which achieves the bound \eqref{4} with equality. Thus we only need to show that $d(C)=5,$ i.e, find a codeword of $C$ of Hamming weight 5. 

Denote $\omega=\alpha^{\frac{n}{3}}$ and $\beta=\alpha^{\frac{n}{q+1}}$. Then $\omega^{3}=1, \beta^{q+1}=1 \Rightarrow \omega^{2}=\omega+1$ and $\beta \in \mathbb{F}_{q^2} \backslash \mathbb{F}_{q}$. Thus there exist $a, b\in \mathbb{F}_{q}^{*}$ such that
\begin{equation}\label{beta}
  \beta^{2}+a\beta+b=0,
\end{equation} i.e.,
$(\frac{\beta}{a})^{2}+\frac{\beta}{a}+\frac{b}{a^{2}}=0$. Thus $x^2+x+\frac{b}{a^{2}}$ is irreducible. By Lemma \ref{Tr}, $\textnormal{Tr}_{\mathbb{F}_{q}/\mathbb{F}_{2}}(\frac{b}{a^{2}})=1$. Since $3 \mid (q+1)$, we can derive that $q=2^{\ell}$ for some odd $\ell$. Thus \[\textnormal{Tr}_{\mathbb{F}_{q}/\mathbb{F}_{2}}(1+\frac{b}{a^{2}})=\textnormal{Tr}_{\mathbb{F}_{q}/\mathbb{F}_{2}}(1)+\textnormal{Tr}_{\mathbb{F}_{q}/\mathbb{F}_{2}}(\frac{b}{a^{2}})=1+1=0.\]
By Lemma \ref{Tr} again, there exists $c_0 \in \mathbb{F}_{q}^{*}$ such that $(\frac{c_0}{a})^{2}+\frac{c_0}{a}+1+\frac{b}{a^{2}}=0$, i.e.,
\begin{equation}\label{ab}
  c^{2}_{0}+ac_0+a^2+b=0.
\end{equation}
 Note that 
 \begin{eqnarray*}
 (\beta+c_0)^{3} &=& \beta^{3}+\beta^{2}c_0+\beta c^{2}_{0}+c^{3}_{0}\\
  &=& \beta(a\beta+b)+(a\beta+b)c_0+\beta c^{2}_{0}+c^{3}_{0} \\
  &=& a(a\beta+b)+b\beta+(a\beta+b)c_0+\beta c^{2}_{0}+c^{3}_{0} \\
  &=& (c^{2}_{0}+ac_0+a^2+b)\beta+ab+bc_0+c^{3}_{0} \\ 
  &=& (a+c_0)b+c^{3}_{0}.
 \end{eqnarray*}
 Thus
 \[(\beta+c_0)^{3}+a^{3}=(a+c_0)(b+c^{2}_{0}+ac_0+a^2)=0,\]
 i.e., $(\frac{a}{\beta+c_0})^{3}=1 \Rightarrow \frac{a}{\beta+c_0}=\omega$ or $\omega^{2}$. Without loss of generality, we suppose $\frac{a}{\beta+c_0}=\omega$ (the proof of the case of $\omega^{2}$ is similar). We prove that $c_0 \neq a$. Otherwise, $c_0=a$, then from Eq. \eqref{ab}, we have $b=a^2$. Then $\beta^{2}+a\beta+a^{2}=0$ by Eq. \eqref{beta}. One can deduce that 
 \[(\frac{\beta}{a})^{3}=1 \Rightarrow \beta^{3(q-1)}=1 \Rightarrow (q+1) \mid 3(q-1) \Rightarrow (q+1) \mid 3,\]
 which contradicts with $q \geq 4$. Now let \[c(x)=c_0+ax^{\frac{n}{3}}+(c_0+a)x^{\frac{2n}{3}}+x^{\frac{n}{q+1}}+x^{\frac{n}{q+1}+\frac{2n}{3}}\]
 which is a polynomial over $\mathbb{F}_{q}$ of weight 5. Note that for $i=0,1,\dots, \frac{n}{3}-1$,
 \[c(\alpha^{3i})=c_0+a+(c_0+a)+\alpha^{\frac{3i n}{q+1}}+\alpha^{\frac{3i n}{q+1}+2i n}=0.\]  
 And 
 \begin{eqnarray*}
 c(\alpha) &=& c_0+a\omega+(c_0+a)\omega^2+\beta+\beta \omega^{2} \\
  &=& c_0+a\omega+(c_0+a)(1+\omega)+\beta+\beta (1+\omega) \\
  &=& a+c_0 \omega+\beta \omega=0.
 \end{eqnarray*}
 Thus $g(x) \mid c(x)$, i.e, $c(x)$ is a codeword of $C$ of weight 5. Thus $d=5$. 
 
 The Theorem is proved.
\end{proof}

\begin{example}
Suppose $q=8$, $m=4$, $n=\frac{q^m-1}{q-1}=585$ and $\alpha \in \mathbb{F}_{8^4}$ is a primitive element. Let  $g(x)=(x^{195}-1)(x-\alpha)(x-\alpha^8)(x-\alpha^{64})(x-\alpha^{256}) \in \mathbb{F}_8[x]$ and $C=\langle g(x) \rangle$. By the Magma software, we can calculate that $d(C)=5$, which concides with Theorem \ref{thm3}.  
\end{example}

\subsection{Constacyclic Perfect LRCs}
For $q$ is even, we have constructed a family of $q$-ary cyclic perfect LRCs in Subsection \ref{cyclic-perfect}. For general $q$, we consider the construction of $q$-ary perfect LRCs via constacyclic codes. Suppose $3 \mid (q+1)$ and $m$ is even. Let $n=\frac{q^m-1}{q-1}.$ Then it can be shown that $3 \mid n$.  Let $\pi$ be a primitive element of $\mathbb{F}_{q^m}$,
\[\theta=\pi^{-q+2} \textnormal{ and } \lambda=\theta^n.\]
Then $\lambda^{q-1}=\theta^{q^m-1}=1$, hence $\lambda \in \mathbb{F}^{*}_q$. Let 
\[\alpha=\pi^{q-1},\]
then $\alpha$ is a primitive $n$-th root of unity and
\[x^n-\lambda=x^n-\theta^n=\prod_{i=0}^{n-1}(x-\theta\alpha^i)\]

\begin{theorem}\label{thm4}
Suppose $q>2$, $3 \mid (q+1)$ and $m$ is even. Let $n=\frac{q^m-1}{q-1}.$
Denote
\[L=\{\theta, \theta\alpha^{3}, \theta\alpha^{6}, \cdots, \theta\alpha^{n-3}\},\]
\[D=\{\theta\alpha, (\theta\alpha)^{q}, \cdots, (\theta\alpha)^{q^{m-1}}\}.\]
Let $g(x)=\prod\limits_{a \in L\bigcup D}(x-a)$ and $C=\langle g(x)\rangle$.
Then $C$ is a $\lambda$-constacyclic $(n=\frac{q^{m}-1}{q-1}, k=\frac{2}{3}n-m, d=5; r=2)$-perfect LRC.
\end{theorem}

\begin{proof}
From the above, we know that the elements of $L$ and $D$ are roots of $x^n-\lambda$, and
\[\prod_{a \in L}(x-a)=\prod_{i=0}^{\frac{n}{3}-1}(x-\theta\alpha^{3i})=x^\frac{n}{3}-\theta^{\frac{n}{3}}.\]
Thus 
\[g(x)=\prod\limits_{a \in L\bigcup D}(x-a)=(x^{\frac{n}{3}}-\theta^{\frac{n}{3}})(x-\theta\alpha)(x-(\theta\alpha)^q)\cdots(x-(\theta\alpha)^{q^{m-1}}).\]
Since $3 \mid (q+1)$, thus $3 \mid (-q+2)$ and $(\theta^{\frac{n}{3}})^{q-1}=(\pi^{\frac{-q+2}{3}})^{q^m-1}=1$, hence $\theta^{\frac{n}{3}} \in \mathbb{F}^{*}_q$. Note that 
\[\theta \alpha=\pi^{-q+2}\pi^{q-1}=\pi.\]
Then $g(x) \in \mathbb{F}_q[x]$ and $g(x) \mid (x^n-\lambda)$ Hence $\dim(C)=\frac{2}{3}n-\deg(g(x))=\frac{2}{3}n-m.$ 
Note that 
\[(\theta\alpha)^q=\pi^q=\pi^{-q+2}\pi^{2(q-1)}=\theta\alpha^{2}.\]
Thus $\theta, \theta \alpha, \theta \alpha^2, \theta \alpha^3$ are roots of $g(x)$. By the BCH bound for constacyclic codes (see Lemma \ref{BCH}), the minimum distance $d(C) \geq 5$.  According to the set $L$, $C$ has locality $r=2$ by Proposition \ref{locality}. Note that $\frac{q^{\frac{r}{r+1}n}}{\frac{rn}{2}(q-1)+1}=\frac{q^\frac{2}{3}n}{q^m}=q^k$, which achieves the bound \eqref{4} with equality. Thus we only need to show that $d(C)=5,$ i.e, find a codeword of $C$ of Hamming weight 5. 

Note that $q>2$ and $3 \mid (q+1)$, thus $q \geq 5$. We can choose $a \in \mathbb{F}^*_q$, such that $a \neq -1, 2$ and $2a-1 \neq 0$ (for example,  
$a=1$ if $q$ is odd; $a \neq 0, -1$, if $q$ is even). Let $\omega=\alpha^{\frac{n}{3}}$, then $\omega^3=1$. Choose $\bar{k}_3 \in \mathbb{Z}$, such that 
\[\pi^{\bar{k}_3}=\frac{a+\omega}{1-\omega}.\]
Let $\bar{k}_3=sn+k_3,$ where $0 \leq k_3 <n$.
We next prove that $\frac{n}{3} \nmid \bar{k}_3$. Otherwise, we have $(\pi^{\bar{k}_3})^{3(q-1)}=1$. Then 
\[\left(\frac{a+\omega}{1-\omega}\right)^{3q}=\left(\frac{a+\omega}{1-\omega}\right)^3\]
\[\Rightarrow \left(\frac{a+\omega^q}{1-\omega^q}\right)^3=\left(\frac{a+\omega}{1-\omega}\right)^3.\]
Since $3 \mid (q+1),$ $\omega^q=\omega^{-1}$. Thus 
\[\left(\frac{a\omega+1}{\omega-1}\right)^3=\left(\frac{a+\omega}{1-\omega}\right)^3\]
\[\Rightarrow (a\omega+1)^3+(a+\omega)^3=0,\]
\[\Rightarrow \omega(\omega+1)(a+1)(a-2)(2a-1)=0,\]
\[\Rightarrow a=-1\textnormal{ or } 2, \textnormal{ or } 2a-1=0,\]
which is a contradiction. Thus $\frac{n}{3} \nmid \bar{k}_3$, hence $k_3 \neq 0, \frac{n}{3}, \frac{2}{3}n$.
Let \[k_4=k_3+\frac{n}{3} (\bmod~n) \textnormal{ with } 0 \leq k_4 <n. \]
Then $k_4 \neq 0, \frac{n}{3}, \frac{2}{3}n$ and $\alpha^{k_4-k_3}=\alpha^{\frac{n}{3}}=\omega \Rightarrow \alpha^{3k_4}=\alpha^{3k_3}$. 
Let $b=\pi^{sn} \in \mathbb{F}^*_q$, then $\pi^{k_3}=\frac{a+\omega}{b(1-\omega)}$. Now, let
\[c(x)=1+a\theta^{-\frac{n}{3}}x^{\frac{n}{3}}-(1+a)\theta^{-\frac{2}{3}n}x^{\frac{2}{3}n}+bx^{k_3}-b\theta^{k_3-k_4}x^{k_4} \in \mathbb{F}_q[x].\]
Then $wt(c(x))=5$, for $i=0,1, \cdots,\frac{n}{3}-1$, we have
\[c(\theta\alpha^{3i})=1+a-(1+a)+b\theta^{k_3}\alpha^{3ik_3}-b\theta^{k_3-k_4}\theta^{k_4}\alpha^{3ik_4}=0,\]
and 
\begin{eqnarray*}
c(\theta\alpha)&=& 1+a\alpha^{\frac{n}{3}}-(1+a)\alpha^{\frac{2}{3}n}+b\pi^{k_3}-b\pi^{k_3}\alpha^{k_4-k_3}\\
&=& 1+a\omega-(1+a)\omega^2+b\pi^{k_3}(1-\omega) \\
&=&-(a+\omega)+b\pi^{k_3}(1-\omega)=0.
\end{eqnarray*}
Hence $g(x) \mid c(x)$, i.e., $c(x)$ is a codeword of $C$ of Hamming weight 5.
\end{proof}

\begin{example}
Suppose $q=5, m=6, n=\frac{q^m-1}{q-1}=3906$ and $\pi$ is a primitive element of $\mathbb{F}_{5^6}$. Let $\theta=\pi^{-3},\lambda=\theta^n=3 \in \mathbb{F}_5$ and $\alpha=\pi^4$. Let $g(x)=(x^{1302}-\theta \alpha)(x-(\theta\alpha))(x-(\theta\alpha)^5)(x-(\theta\alpha)^{25})(x-(\theta\alpha)^{125})(x-(\theta\alpha)^{625})=x^{1308} + x^{1306} + 4x^{1305} + x^{1304} + 2x^{1302} + 2x^6 + 2x^4 + 3x^3 + 2x^2 +4 \in \mathbb{F}_5[x]$ and $C=\langle g(x) \rangle$. With the help of Magma, we can calculate that $d(C)=5$, which coincide with Theorem \ref{thm4}.
\end{example}

\section{Conclusion}
In this paper, we consider new constructions of Singleton-optimal LRCs and perfect LRCs via cyclic and constacyclic codes. Firstly, we obtain two classes of $q$-ary cyclic $(n, k, d=6; r=2)$-Singleton-optimal LRCs with $n=3(q+1)$ and $n=\frac{3(q+1)}{2}$, respectively. To the best of our knowledge, this is the first construction of $q$-ary Singleton-optimal LRCs of length $n>q+1$ and minimum distance $d \geq 5$ via cyclic codes. Secondly, by using cyclic and constacyclic codes, we give two constructions of perfect LRCs.  Cyclic and constacyclic codes have nice algebraic structures and efficient  encoding and decoding algorithm. Thus it is interesting to find more new constructions of Singleton-optimal LRCs and perfect LRCs with longer code length via cyclic and constacyclic codes, and investigate the theoretical bounds for cyclic and constacyclic LRCs in the future.

\section*{Ackonwledgements}
This research is supported in part by National Key Research and Development Program of China under Grant Nos. 2022YFA1004900, 2021YFA1001000, 2018YFA0704703 and 2022YFA1005001, the National Natural Science Foundation of China under Grant Nos. 62201322, 62171248, 12141108, 61971243 and 12226336, the Natural Science Foundation of Shandong (ZR2022QA031), the Natural Science Foundation of Tianjin (20JCZDJC00610), the Fundamental Research Funds for the Central Universities, Nankai University, and the Nankai Zhide Foundation.

\end{document}